\documentclass{llncs}

\newif\ifdoc

\usepackage{hyperref}

\usepackage{graphicx}
\usepackage{color}

\usepackage{colortbl}

\definecolor{grey}{rgb}{.85,.85,.85}

\usepackage{amsmath}
\usepackage{amssymb}
\usepackage{amsfonts}
\usepackage{mathtools}
\usepackage{algorithm}
\usepackage{algorithmicx}

\ifdoc
\renewcommand{\mathbb}{{}}
\renewcommand{\boldsymbol}[1]{{#1}}
\renewcommand{\mathcal}[1]{{#1}}
\renewcommand{\binom}[2]{{#1\choose#2}}
\renewcommand{\operatorname}[1]{{#1}}
\fi

\usepackage{algpseudocode}

\newcommand{\Proba}{\operatorname{Pr}}
\newcommand{\F}{\mathbb F}
\newcommand{\N}{\mathbb N}
\newcommand{\ev}{\operatorname{ev}}
\newcommand{\set}[1]{\left\{#1\right\}}
\newcommand{\size}[1]{\left|#1\right|}
\newcommand{\vect}[1]{\left(#1\right)}
\renewcommand{\refeq}[1]{Eq.~\ref{#1}}
\newcommand{\reffig}[1]{Fig.~\ref{#1}}

\newcommand{\MRM}{\text{Mult}}
\newcommand{\RM}{\text{RM}}
\newcommand{\RS}{\text{RS}}

\newcommand{\restrict}[2]{#1_{#2}}
\newcommand{\Hasse}[2]{H(#1,#2)}
\newcommand{\coeff}[2]{\text{coeff}(#1,#2)}
\newcommand{\vX}{{\boldsymbol X}}

\newcommand{\vZ}{{\boldsymbol Z}}
\newcommand{\vi}{{\boldsymbol i}}
\newcommand{\vj}{{\boldsymbol j}}

\newcommand{\vv}{{\boldsymbol v}}
\newcommand{\vV}{\boldsymbol V}
\newcommand{\vU}{\boldsymbol U}
\newcommand{\vQ}{{\boldsymbol Q}}
\newcommand{\vP}{{\boldsymbol P}}
\newcommand{\vR}{{\boldsymbol R}}

\begin{document}

\title{A Storage-efficient and Robust Private Information Retrieval
  Scheme allowing few servers}
%
%
%
%
%

%
\author{
%
%
 Daniel Augot\inst{1,2},
 Fran\c coise Levy-dit-Vehel\inst{1,2,3},
 Abdullatif Shikfa\inst{4}
 }
\institute{
INRIA 
\and 
Laboratoire d'informatique de l'\'Ecole polytechnique
\and   
ENSTA ParisTech/U2IS
\and
Alcatel-Lucent 
}
\maketitle
\begin{abstract}
  Since the concept of locally decodable codes was introduced by Katz
  and Trevisan in 2000 \cite{KatzTrevisan00}, it is well-known that
  information theoretically secure private information retrieval
  schemes can be built using locally decodable codes
  \cite{YekhaninBookLDCPIR:2010}.  In this paper, we construct a
  Byzantine robust PIR scheme using the multiplicity codes introduced
  by Kopparty {\em et al.}  \cite{KopSarYek2011}. Our main
  contributions are on the one hand to avoid full replication of the
  database on each server; this significantly reduces the global
  redundancy. On the other hand, to have a much lower locality in the
  PIR context than in the LDC context. This shows that there exists
  two different notions: LDC-locality and PIR-locality. This is made
  possible by exploiting geometric properties of multiplicity codes.
\end{abstract}




\section{Introduction}
Private information retrieval allows a user to privately retrieve a
record of a database, in the sense that the database server does not
know which record the user is asking for. The applications of this
functionality are numerous.  Imagine for instance doctors having to
query a company-wide database storing medical for patients, or a
police officer wanting to request financial data from the fiscal
administration. In both cases, to respect privacy of the patient, or
secrecy of the inquiry, it is desirable that the central
administration does not know about the queries sent by these users
(the doctor or the police officer). A private information retrieval
protocol will allow these users to send their queries to the
databases, without revealing what they are asking for (either the name
of patient, or the name of the suspect under inquiry). Another example
is an Internet user who wants to use cloud-based remote storage
services, like DropBox, GoogleDrive, CloudMe, hubiC, etc, to store
data, and retrieve portion of its data without revealing to these
remote services anything about what he is after.

\paragraph{Related work.}
The problem of Private Information retrieval (PIR) was introduced in
1995 by Chor, Goldreich, Kushilevitz and Sudan~\cite{CGKS95}. A PIR
protocol is a cryptographic protocol the purpose of which is to
protect the privacy of a user accessing a public database via a
server, in the sense that it makes it possible for a user to query a
particular record of the database without revealing to the server
which record he wants to retrieve.  We here deal with {\em information
  theoretic} PIR, as opposed to {\em computationally secure}
PIR~\cite{KO97}. In an \emph{information theoretic} PIR setting, a
server gets no information about the identity of the record of user
interest even if it has unlimited computing power: the queries sent to
the server must not be correlated to the actual record the user is
looking for.  In~\cite{CGKS95} it is shown that when accessing a
database located on a single server, to completely guarantee the
privacy of the user in an information theoretic sense, one needs to
download the entire database, which results in a communication
complexity of $O(N)$, $N$ being the bit-size of the database.  Thus
scenarios have been introduced where the database is replicated across
several, say $\ell$, servers, and the proposed schemes have
communication complexity $O(N^{1/\ell})$, for $\ell\geq 3$.  Such
multiple-server settings have been investigated since then, and the
best communication complexity to date is $N^{O(1/(\log_2 \log_2N))}$
for 3-server PIR protocols (from matching vector codes
construction~\cite{Yekhanin08,Efremenko09}) and
$N^{O((\log_2\log_2\ell)/\ell \log_2\ell)}$ for $\ell \geq
3$~\cite{Beimeletal2001} .

Beimel and Stahl \cite{SCN02,JOC07} have proposed several robust
information theoretic PIR protocols, based on polynomial
interpolation, as well as on Shamir's secret sharing scheme. They have
built a generic transformation from regular to robust PIR protocols
that relies on perfect hash families. They also addressed the
Byzantine setting.  Recently, Devet, Goldberg and Heninger
\cite{DeGoHe2012} proposed an Information-Theoretic PIR tolerating the
maximum possible number of Byzantine servers. 
In all these previous proposals, the (encoded or not) database is fully replicated among the servers.

\paragraph{Our contribution.} Our main concern is to reduce the global
storage overhead. We achieve this by avoiding full replication of the
database among the servers. We use multiplicity codes and exploit the
geometry of $\F_q^m$ to partition the encoded database (codeword) of
bit-size $N$ into $q$ shares of equal size, and distribute them among
the servers (one share for one server). This way, we reduce the
storage on each server from $N$ bits down to $N/q$ bits, $q$ being the
number of servers, while totally preserving the information theoretic
security of the PIR protocol.  Here $N=\log_2 (q^{\sigma q^m})=\sigma
q^m \log_2 q$, with $\sigma=\binom{m+s-1}{m}$, and $s$ is an auxiliary
small integer (say $s \leq 6$) used in the construction of
multiplicity codes. Given that the code has rate $R$, the storage
overhead of our scheme is thus $\frac 1R$ instead of $\frac1R\ell$ for
schemes with full replication of the encoded database (as in the
standard LDC to PIR reduction), $\ell$ being the number of servers
($\ell=q$ in our scheme).  The number of servers is also drastically
reduced, from $\sigma(q-1)$ to $q$, see~Fig~\ref{fig:params}.

The communication complexity in bits (total number of bits sent by the
user to all the servers as queries of our protocol) is $(m-1)q\sigma
\log_2 q$, and the total number of bits answered by the servers is $q
\sigma^2 \log_2 q$. Thus the communication complexity is
$(m-1+\sigma)q\sigma \log_2q$ bits. Putting $\ell=q$ the number of
servers, and in contexts where $s$ is small, say $s\leq 6$, this gives
a communication complexity of $O(\ell(\log_2 N)^s)$.

Our protocol tolerates $\nu=\lfloor t \rfloor$ byzantine servers, $t=1/2(q-1-d/s)$, $d$ being the degree of the multiplicity code, in the sense that even if $\nu$ out of $q$ servers always answer wrongly, then the database item can still be correctly recovered by the user. Thus 
our protocol is a $\nu$-Byzantine robust PIR protocol.
The property of being robust is a built-in feature of the decoding algorithms that are involved in the process of retrieving the database item. 


\paragraph{Organization of the paper.}
In section~\ref{prelim} we recall the basics of locally decodable and
self-correctable codes, private information retrieval schemes, and the
link between the two notions; we also set the necessary material and
notation to define multiplicity codes, namely Hasse
derivatives. Section~\ref{codes} describes the multiplicity codes
\cite{KopSarYek2011} as a generalization of Reed Muller codes, and
explains their local decoding. Section~\ref{hyp} contains our main
ideas: we explain how we use multiplicity codes in a PIR scenario in
such a way as to avoid full replication of the encoded database. We
also explain how we achieve the Byzantine robustness property of our
protocol.  We end the paper by numerical tables showing the main
features of the codes (rate, locality) for various parameter sizes.

\section{Preliminaries} \label{prelim}

\subsection{Locally decodable and locally self-correctable codes}

A code in the ambient space is seen  as an encoding map, which
encodes a message of $k$ symbols on an alphabet $\Delta$ into
code-vectors, or {\em codewords} of $n$ symbols on some alphabet
$\Sigma$ (possibly different from $\Delta)$. That is, it is a
one-to-one map $C:\Delta^k\rightarrow\Sigma^n$. The decoding problem
is to find codewords close enough to any element $y$ in the ambient
space (the ``received word'' in coding theory language). Formally,
given a distance $d()$, code $C\subset \Sigma^n$, for a given
$y=(y_1,\dots,y_n) \in \Sigma^n$, one has to find one, some, or all
codewords $c\in C$ such that $d(c,y)$ is small. In our setting, the
distance $d(x,y)$ is the Hamming distance which is the number of
indices $i$ where $x_i\neq y_i$. A major concern is to build codes
with small redundancy, or equivalently, large \emph{rate}, where the
rate is $(k\log|\Delta|)/(n\log|\Sigma|)$. In classical settings,
$\Delta=\Sigma$, and the rate is simply $k/n$.

%

Locally decodable codes, in short LDCs, allow efficient sublinear time
decoding. More precisely, an $\ell$-query LDC allows to
probabilistically recover any symbol of a message by looking at only
$\ell \leq k$ randomly chosen coordinates of its - possibly corrupted
- encoding. The major objective is to have $\ell\ll k$.  Although LDCs
appeared in the PCP literature in early 90's
\cite{YekhaninBookLDCPIR:2010}, their first formal definition is due
to Katz and Trevisan in 2000 \cite{KatzTrevisan00}.  The number $\ell$
of queried symbols is the {\em query complexity}, that we also call
here {\em locality}.  Formally: 
\begin{definition}
A code $C: \Delta^k \rightarrow \Sigma^n$ is
$(\ell,\delta)${\em -locally decodable}
if there exists a randomized decoding algorithm
${\mathcal A}$ such that 
\begin{enumerate}
\item for any \emph{message} $x \in \Delta^k$ and any $y \in \Sigma^n$ with
  $d(C(x),y) < \delta n$, we have, for all $i \in [k]$,
$
\Proba[{\mathcal A}^{y}(i)=x_i] \geq \frac{2}{3}
$,
\item ${\mathcal A}$ makes at most $\ell$ queries to $y$.
\end{enumerate}
\end{definition}
Here, and in the following, ${\mathcal A}^{y}$ means that ${\mathcal
  A}$ is given query access to $y$, and the probability is taken over
all internal random coin tosses of ${\mathcal A}$.  In the case when
one wants to probabilistically recover {\em any} codeword symbol and
not only information symbols, one has the following definition.
\begin{definition} A code $C: \Delta^k \rightarrow \Sigma^n$ is
$(\ell,\delta)${\em -locally self-correctable (LCC)} if there exists a randomized decoding algorithm ${\mathcal A}$ such that 
\begin{enumerate}
\item for any \emph{codeword} $c \in \Sigma^n$ and $y \in \Sigma^n$
  with $d(c,y) < \delta n$, we have, for all $i \in [k]$,
$
\Proba[{\mathcal A}^{y}(i)=c_i] \geq \frac{2}{3}$,
\item  ${\mathcal A}$ makes at most $\ell$ queries to $y$.
\end{enumerate}
\end{definition}

When $\Delta=\Sigma=\F_q$, the finite field with $q$ elements, and
when the code is $\F_q$-linear, one can easily construct an LDC from a
LCC \cite{YekhaninBookLDC:2012}.  No known constructions of LDCs or
LCCs minimize both $\ell$ and the length $n$ simultaneously.  The
oldest class of LDCs are the Reed-Muller codes over $\F_{q}$, whose
codewords are the evaluations of $m$-variate polynomials of total
degree at most $d$ over $\F_{q}$ on all the points of $\F_{q} ^m$. The
main issues are thus to minimize one parameter given that the other
one is fixed. With this respect, constructions of subexponential
length codes with constant query complexity $\ell \geq 3$
exist~\cite{YekhaninBookLDCPIR:2010}. On the other side, constant rate
LDCs feature an $\ell$ which is known to lie between $\Omega(\log_2k)$
and $\Theta(k^{\epsilon})$, with explicit constructions for the latter
bound. A major result is the construction of high-rate (i.e. $>1/2$)
locally self-correctable codes with sublinear query complexity, in
the presence of a constant (as a function of the distance of the code)
fraction of errors. Those codes are known as {\em Multiplicity Codes}
and were introduced by Kopparty, Saraf and Yekhanin in 2011
\cite{KopSarYek2011}. They generalize the Reed-Muller codes by
evaluating high degree multivariate polynomials as well as their
partial derivatives up to some order $s$.  Using high-degree
polynomials improves on the rate, while evaluating their partial
derivatives compensates for the loss in distance.  Other LDC
constructions achieving rate $>1/2$ and query complexity
$n^{\epsilon}$ are the one of Guo \emph{et al.}~\cite{GKS13} based on
lifting affine-invariant codes (namely, Reed-Solomon codes), and the
Expander codes of Hemenway \emph{et al.}~\cite{HOW13}.

In this work, we  use Multiplicity codes, but recall Reed-Muller
codes and their local decoding for the sake of comprehension. These
codes provide the simplest geometric setting for partitioning a
codeword and laying it out on servers. We think such a partition can
be done for other families of LDC codes, e.g.\ matching-vector codes,
affine invariant codes and possibly Expander codes.

\subsection{Private information retrieval schemes}
We model the database as a string $x$ of length $k$ over $\Delta$. An
$\ell$-server PIR scheme involves $\ell$ servers
$S_1,\ldots,S_{\ell}$, each holding the same database $x$, and a user
who knows $k$ and wants to retrieve some value $x_i$, $i \in [k]$,
without revealing any information about $i$ to the servers.
\begin{definition}[Private Information Retrieval (PIR)]
An $\ell$-server $p$-PIR protocol is a triple  $({\mathcal Q},{\mathcal
  A},{\mathcal R})$ of algorithms running as follows:
\begin{enumerate}
\item User obtains a random string $s$; then he invokes ${\mathcal Q}$
  to generate an $\ell$-tuple of queries
  $(q_1,\ldots,q_{\ell})={\mathcal Q}(i,s)$. 
\item For $1 \leq j \leq \ell$,
  User sends $q_j$ to server $S_j$;
\item Each $S_j$ answers $a_j={\mathcal A}(j,x,q_j)$ to User; 
\item User recovers $x_i$ by applying the reconstruction algorithm
${\mathcal R}(a_1,\ldots,a_{\ell},i,s)$.
\end{enumerate}
Furthermore the protocol has the \emph{Correctness property}: for any
$ x \in \Delta^k$, $i \in [k]$, User recovers $x_i$ with probability at
least $ p$; and the \emph{Privacy property}: each server individually
can obtain no information about $i$. 
\end{definition}
The Privacy property can be obtained by requiring that for all $ j \in
[\ell]$, the distribution of the random variables ${\mathcal
  Q}(i,\cdot)_j$ are identical for all $i \in [k].$ Katz and Trevisan
\cite{KatzTrevisan00}, introduced a notion very relevant in the
context of locally decodable codes: that of {\em smooth codes}. The
notion of smooth codes captures the idea that a decoder cannot read
the same index too often, and implies that the distributions ${\cal
  Q}(i,\cdot)_j$ are close to uniform. All known examples are such
that the distribution ${\cal Q}(i,\cdot)_j$ are actually
uniform. Uniform distribution of the queries among codeword (or
received word) coordinates is what is needed in the PIR setting in
order to achieve information theoretic privacy of the queries.  The
locality as a core feature of LDCs, together with the fact that in all
known constructions of LDCs the queries made by the local decoding
algorithm ${\mathcal A}$ are uniformly distributed, make the
application of LDCs to PIR schemes quite natural. Note also that
conversely PIR schemes can be used to build LDCs with best asymptotic
code-lengths~\cite{Beimeletal2001,Yekhanin08,Efremenko09}. The lemma below describes how it formally works.

\begin{lemma}[Application of LDCs to PIR schemes] \label{lemma:LDC2PIR}
Suppose there exists
an $\ell$-query locally decodable code $C: \Delta^k \rightarrow
\Sigma^n$, in which each decoder's query is uniformly distributed over
the set of codeword coordinates.  Then there exists an $\ell$-server
1-PIR protocol with $O(\ell (\log_2n+\log_2\size{\Sigma}))$ communication to
access a database $x \in \Delta^k$.
\end{lemma}
\begin{proof} Given an LDC $C: \Delta^k \rightarrow \Sigma^n$ as in
  the lemma, one constructs the following PIR protocol. First, in a
  preprocessing step, for $1 \leq j \leq \ell$, server $S_j$ encodes
  $x$ with $C$. Then, to actually run the protocol, User tosses random
  coins and invokes the local decoding algorithm to determine the
  queries $(q_1,\ldots,q_{\ell}) \in [n]^{\ell}$ such that $x_i$ can
  be computed from $\{C(x)_{q_j}\}_{1 \leq j \leq \ell}$.  For $1 \leq
  j \leq \ell$, User sends $q_j \in [n]$ to server $S_j$, and each
  server $S_j$ answers $C(x)_{q_j}\in \Sigma$. Finally, User applies the local
  decoding algorithm of $C$ to recover $x_i$.

  This protocol has the communication complexity claimed in the
  lemma. Furthermore, as the user applies the local decoding algorithm
  with non corrupted inputs $\{C(x)_{q_j}\}_{1 \leq j \leq \ell}$, he
  retrieves $x_i$ with probability 1. Uniformity of the distribution
  of the decoder's queries over $[n]$ ensures the
  information-theoretic privacy of the protocol.
\end{proof}
\subsection{Hasse derivative for multivariate polynomials}
\subsubsection{Notation}
Considering $m$ indeterminates $X_1,\dots,X_m$, and $m$ positive
integers $i_1,\dots,i_m$, we use the short-hand notation
\begin{align*}
\vX&=(X_1,\dots,X_m)&
\vX^\vi&=X_1^{i_1}\cdots X_m^{i_m},
&\F_q[\vX]&=\F_q[X_1,\dots,X_m]\\
\vi&=(i_1,\dots,i_m)\in\N^m&\size{\vi}&=i_1+\dots+i_m&
\vP&=(p_1,\dots,p_m)\in\F_q^m
\end{align*}
i.e.\ we use bold symbols for vectors, points, etc, and standard
symbols for uni-dimensional scalars, variables, etc. In general, we
write polynomials $Q\in\F_q[\vX]=\F_q[X_1,\dots,X_m]$ without
parenthesis and without variables, and $Q(\vX)$ (resp. $Q(\vP)$) when
the evaluation on indeterminates (resp. points) has to be specified.
For $\vi, \vj \in \N^m$, $\vi \gg \vj$ means $i_t \geq j_t \,\forall 1 \leq t \leq m$.

\subsubsection{Hasse derivative}
Given a multi-index $\vi$, and $F\in\F_q[\vX]$, the $\vi$-th Hasse
derivative of $F$, denoted by $\Hasse F {\vi}$, is the
coefficient of $\vZ^{\vi}$ in the polynomial $F(\vX+\vZ) \in
\F_q[\vX,\vZ]$, where $\vZ=(Z_1,\ldots,Z_m)$. More specifically, let
$F(\vX)=\sum_{\vj\gg0}f_{\vj}\vX^{\vj}$, then
\begin{gather*}
\begin{aligned}
  F(\vX+\vZ)&=\sum_{\vj}f_{\vj}(\vX+\vZ)^{\vj}=\sum_{\vi}\Hasse{F}{\vi}(\vX)\vZ^{\vi},\\
\end{aligned}
\shortintertext{where $\vZ^\vi$ stands for $Z_1^{i_1}\cdots
  Z_m^{i_m}$, and }
\begin{aligned}
  \Hasse
  F{\vi}(\vX)&=\sum_{\vj\gg\vi}f_{\vj}\binom{\vj}{\vi}\vX^{\vj-\vi}\quad
  \text{with } \binom \vj\vi=\binom {j_1}{i_1}\cdots\binom{j_m}{i_m}.
\end{aligned}
\end{gather*}
Considering a vector $\vV\in\F_q^m\setminus\set{0}$, and a base point
$\vP$, we consider the restriction of $F$ to the line
$D=\set{\vP+t\vV:\; t\in \F_q}$,
which is a univariate polynomial that we denote by $\restrict F
{\vP,\vV}(T)=F(\vP+T\vV)\in \F_q[T]$. We have the following relations:
\begin{align}
  \restrict F {\vP,\vV}(T)&=\sum_{\vj} \Hasse F {\vj} (\vP) \vV^{\vj} T^{\size{\vj}},\\
  \coeff {\restrict F{\vP,\vV}} i&=\sum_{\size \vj=i} \Hasse F
  {\vj}(\vP)\vV^{\vj},\label{coeffeq}
\\
  \Hasse {\restrict F {\vP,\vV}} i(\alpha)&=
  \sum_{\size{\vj}=i}\Hasse F
  {\vj}(\vP+\alpha\vV)\vV^{\vj}, \quad \text{for all }
  \alpha\in\F_q\label{HasseUniv} 
\end{align}

\section{Multiplicity codes} \label{codes}
\subsection{Local decoding of Reed-Muller codes}
We enumarte the finite field $\F_q$ with $q$ elements as
$\F_q=\set{\alpha_0=0,\alpha_1,\dots,\alpha_{q-1}}$.
We denote by $\F_q[\vX]_d$ the set of polynomials of degree less than or equal to $d$, which has dimension $k=\binom {m+d}d$. We enumerate all the points
in $\F_q^m$:
\begin{equation}\label{eq:enum:Fqm}
\F_q^m=\left\{\vP_1,\dots,\vP_n\right\}
\end{equation}
where $\vP_i=\vect{P_{i,1},\dots,P_{i,m}}\in\F_q^m$, is an $m$-tuple of
$\F_q$-symbols, and $n=q^m$.  We encode a polynomial $F$ of degree $\leq d$ into a codeword $c$ of length $n$ using the evaluation
map
\[
\ev:\begin{array}[t]{rcl}
\F_q[\vX]_d&\rightarrow&\F_q^n\\
F & \mapsto & \left(F(\vP_1),\dots,F(\vP_n)\right)
\end{array}
\]
and the $d$-th order Reed-Muller code is $\RM_d=\left\{\ev(F)\;\mid
  F\in \F_q[\vX]_d\right\}$.  The evaluation map $\ev$ encodes $k$ symbols
into $n$ symbols, and the rate is $R=k/n\in[0,1]$. A codeword $c\in
RM_d$ can be indexed by integers as $c=(c_1,\dots,c_n)$ or by points
as $c=(c_{\vP_1},\dots,c_{\vP_n})$, where $c_i=c_{\vP_i}$.

Assuming $d<q$, we now recall how $\RM_d$ achieves a locality of
$\ell=q-1$ as follows. Suppose that $c=\ev(F)\in \RM_d$ is a codeword,
and that $c_j=c_{\vP_j}$ is looked for. Then, the local decoding
algorithm randomly picks a non-zero vector $\vV\subset \F_q^m\setminus
\set 0$ and considers the line $D$ of direction $\vV$ passing through
$\vP_j$:
\begin{align*}
  D&=\left\{\vP_j+t\cdot \vV\; \mid t\in
    \F_q\right\}
  =\left\{\vP_j+0\cdot \vV,\vP_j+\alpha_1\cdot \vV, \dots, \vP_j+\alpha_{q-1}\cdot \vV\right\}\\
  &=\left\{\vR_0=\vP_j,\dots,\vR_{q-1}\right\}\subset \F_q^m.
\end{align*}
Then, the points $\vR_1,\dots,\vR_{q-1}$ are sent as queries, and the
decoding algorithm receives the answer:
\[
\left(y_{\vR_1},\dots,y_{\vR_{q-1}}\right)\in \F_q^{q-1}.
\]
In case of no errors,
$\left(y_{\vR_1},\dots,y_{\vR_{q-1}}\right)=\left(c_{\vR_1},\dots,c_{\vR_{q-1}}\right)$. Now
\begin{align*}
  c_{\vR_u}&=F(\vP_j+\alpha_u\cdot \vV)=\restrict
  F{\vP,\vV}(\alpha_u), \; u=1,\dots,q-1,
\end{align*}
where 
\begin{equation}
\restrict F{\vP,\vV}=F(\vP+T\cdot \vV)\in\F_q[T]\label{eq:defFPV}
\end{equation}
is the restriction of $F$ to the line $D$, which is a univariate
polynomial of degree less than or equal to $d$. That is,
$\left(c_{\vR_1},\dots,c_{\vR_{q-1}}\right)$ belongs to a Reed-Solomon
code $\RS_d$ of length $q-1$ and dimension $d+1$. In case of
errors, $\left(y_{\vR_1},\dots,y_{\vR_{q-1}}\right)$ is a noisy
version of it. Using a decoding algorithm of $\RS_d$, one can recover
$F_{\vP,\vV}$, and then $c_{\vP_j}$ is found as
$c_{\vP_j}=F_{\vP,\vV}(0)$. 

The main drawback of these codes is the condition $d<q$, which imposes
a dimension $k=\binom {d+m}{m} <\binom{q+m}m\sim q^m/m!$. For a fixed
alphabet $\F_q$, the rate $R=k/q^m<1/m!$ goes to
zero very fast when the codes get longer.

\subsection{Multiplicity codes and their local decoding}\label{decloc}
To obtain codes with higher rates, we need a derivation order $s>0$
and an extended notion of evaluation. There are
$\sigma=\binom{m+s-1}{m}$ Hasse derivatives $\Hasse F\vi$ of a
polynomial $F$ for multi-indices $\vi$ such that $\size\vi<s$.
Letting $\Sigma=\F_q^\sigma$, we generalize the evaluation map at a
point $\vP$:
\[
\ev^s_\vP:\begin{array}[t]{rcl}
  \F_q[\vX]&\rightarrow& \F_q^{\sigma}\\
   F&\mapsto &\left(\Hasse F \vv(\vP)\right)_{\size \vv < s}
\end{array}
\]
and, given an enumeration of the points as in \refeq{eq:enum:Fqm}, the
total evaluation rule is
\[
\ev^s:\begin{array}[t]{rcl}
  \F_q[\vX]&\rightarrow& \Sigma^n\\
F&\mapsto
  & \left(\ev^s_{\vP_1}(F),\dots,\ev^s_{\vP_n}(F)\right).
\end{array}
\]
Given $y=\ev^s_\vP(F) \in\Sigma$, we denote by $y_{\vv}$ the
coordinate of $y$ corresponding to the $\vv$-th derivative of $F$.  As
in the case of classical Reed-Muller codes, we denote by
$(c_1,\ldots,c_n)=(c_{\vP_1},\ldots,c_{\vP_n})=\ev^s(F)$,
i.e. $c_i=c_{\vP_i}=\ev^s_{\vP_i}(F)$.  We can consider $\F_q[\vX]_d$,
with $ d<s(q-1)$~\cite{KopSarYek2011}, and the corresponding code is
\[
\MRM^s_{d}=\left\{\ev^s(F)\;\mid F\in \F_q[\vX]_d\right\}.
\]
Using the language of locally decodable codes, we have a code
$\MRM^s_d:\Delta^k\rightarrow \Sigma^n$, with $\Delta=\F_q$, and
$\Sigma=\F_q^\sigma$. The code $\MRM^s_d$, is a $\F_q$-linear space,
whose dimension over $\F_q$ is $k=\binom{m+d}d$.  Its rate is
$R=(\log_q \size{\F_q[X]_d})/(\log_q \size{\Sigma^n})={k}/({\sigma n})=\binom{m+d}m/\left(\binom{m+s-1}{m}\cdot q^{m}\right).$
Its minimum distance is (from Generalized Schwartz-Zippel Lemma)
$q^m-\frac dsq^{m-1}$. 

This family of codes has a locality of
$(q-1)\sigma=(q-1)\binom{m+s-1}{m}$ queries. Here is how the local
decoding algorithm works. Let $j$ be the index of the point where we
want to local decode, i.e. $c_j=c_{\vP_j}$ is looked for.  The
algorithm randomly picks $\sigma$ vectors $\vU_i\in
\F_q^m\setminus\set{0}$, $i=1,\dots,\sigma$. For each $\vU_i$,
$i=1,\dots,\sigma$, consider the line of direction $\vU_i$ passing
through $\vP_j$:
\begin{align*}
  D_{i}
  &=\left\{\vP_j+0\cdot \vU_i,\vP_j+\alpha_1\cdot \vU_i, \dots, \vP_j+\alpha_{q-1}\cdot \vU_i\right\}\\
  &=\left\{\vR_{i,0}=\vP_j,\vR_{i,1},\dots,\vR_{i,q-1}\right\}\subset \F_q^m
\end{align*}
For each $i$, $1 \leq i \leq \sigma$, the algorithm queries the
received word at points $\vR_{i,1},\dots,\vR_{i,q-1}$, and gets the answers
\[
\left(y_{\vR_{i,1}},\dots,y_{\vR_{i,q-1}}\right)\in\Sigma^{q-1},
\]
thus a total of $(q-1)\sigma$ queries in $\F_q^m$, and $\sigma(q-1)$
answers from $\Sigma$. In case of no errors, we have
\[(y_{\vR_{i,b}})_\vv=\Hasse{F}{\vv}(\vR_{i,b}),
\quad b=1,\dots,q-1,
\]
where $(y_{\vR_{i,b}})_\vv$ is the $\vv$-th coordinate of
$y_{\vR_{i,b}}$, and, using \refeq{HasseUniv}, we can compute
\begin{equation}\label{eq:noerrors}
  \Hasse{\restrict{F}{\vP_j,\vU_i}}{e}(\alpha_b)=\sum_{\size \vv=e}\Hasse{F}{\vv}(\vR_{i,b})\vU_i^{\vv}
  \quad \left\{\begin{array}{l}1\leq b\leq q-1,\\  0\leq e <s\end{array}\right.
\end{equation}
Having the values $\Hasse{\restrict{F}{\vP_j,\vU_i}}{e}(\alpha_b)$, for $1 \leq b \leq q-1$ and
$\size{\vv}<s$, we can then recover $\restrict F{\vP_j,\vU_i}$ by
Hermite interpolation. Next we solve, for the indeterminates $\Hasse F
{\vv}(\vP_j)$, $\size\vv<s$, the linear system derived
from~\refeq{coeffeq}:
\[
\coeff {\restrict F{\vP_j,\vU_i}} e=\sum_{\size \vv=e} \Hasse F
{\vv}(\vP_j)\vU_i^{\vv}\quad \left\{\begin{array}{l}
e=0,\dots,s-1,\\ i=1,\dots\sigma,
\end{array}\right.
\]
and we output 
$\left\{\Hasse F {\vv}(\vP_j), \;\size{\vv}<s\right\}=\ev^s_{\vP_j}(F)$.

In case of errors, for each direction $\vU_i$, we define a function
$h_i: \F_q^* \rightarrow \F_q^{\{0,\ldots,s-1\}}$, $\alpha_b \mapsto
h_{i}(\alpha_b)$, such that
\begin{equation}\label{eq:h_i:alpha_n:e}
(h_{i}(\alpha_b))(e)=\sum_{\size \vv=e}(y_{\vR_{i,b}})_\vv\vU_i^{\vv},
\quad \left\{\begin{array}{l}1\leq b\leq q-1\\  0\leq e <s\end{array}\right.
\end{equation}
By virtue of~\refeq{eq:noerrors}, note that
$h_{i}(\alpha_b)(e)$ is  the (erroneous) $e$-th Hasse derivative
of $\restrict{F}{\vP_j,\vU_i}$ at $\alpha_b$.

Having $h_{i}(\alpha_b)(e)$ for all $e \in \{0,\ldots,s-1\}$ and all
$b \in \{1,\ldots,q-1\}$, $\restrict{F}{\vP_j,\vU_i}$ is recovered
using a decoding algorithm of univariate multiplicity codes
(see~\cite{KopSarYek2011}), provided
$d(\ev^s(\restrict{F}{\vP_j,\vU_i}),h_i) \leq
\frac{(q-1)-d/s}{2}$. Once we have recovered $\restrict
F{\vP_j,\vU_i}$, we solve for the indeterminates $\Hasse F
{\vv}(\vP_j)$, $\size\vv<s$, the linear system derived
from~\refeq{coeffeq}:
\begin{equation}\label{eq:erroneous:lin}
\coeff {\restrict F{\vP_j,\vU_i}} e=\sum_{\size \vv=e} \Hasse F
{\vv}(\vP_j)\vU_i^{\vv}  \quad \left\{\begin{array}{l}
t=0,\dots,s-1,\\ i=1,\dots\sigma
\end{array}\right.
\end{equation}
and we output $ \left\{\Hasse F {\vv}(\vP_j),\; \size{\vv}<s\right\}=\ev^s_{\vP_j}(F) $. This local decoding
algorithm is sketched in Alg~\ref{algo:localMRM}. In case of more than
$\frac{(q-1)-d/s}{2}$ errors in some directions, the linear
system~\ref{eq:erroneous:lin} may have erroneous equations. In this
case, due to lack of space, we refer the reader
to~\cite{KopSarYek2011}.

\begin{algorithm}
\begin{algorithmic}[1]
  \Require Oracle Access to $y=\vect{y_1,\dots,y_n}$, a noisy version
  of $c=\ev^s(F)\in\MRM_d$.
  \renewcommand{\algorithmicrequire}{\textbf{Input:}} \Require $j\in
  [n]$, the index of the symbol $c_j$ looked for in $c$
  \renewcommand{\algorithmicrequire}{\textbf{Output:}} \Require
  $c_j=c_{\vP_j}=\ev^s_{\vP_j}(F)$ 

  \State Pick distinct $\sigma$ non zero random vectors
  $\vU_1,\dots,\vU_\sigma$ giving $\sigma$ different lines
 \For{i=1 to $\sigma$} \label{pick_lines}

  \State Consider the line
\[
  D_i
  =\left\{\vP_j+0\cdot \vU_i, \vP_j+\alpha_1\cdot \vU_i, \dots, \vP_j+\alpha_{q-1}\cdot \vU_i\right\}=\left\{\vR_{i,0},\dots,\vR_{i,q-1}\right\}
\]
\State Send $\vR_{i,1},\dots,\vR_{i,q-1}$, as queries,
\State\label{identify} Receive the answers: $y_{\vR_{i,1}},\dots,y_{\vR_{i,q-1}}$, $y_{\vR_{i,b}} \in \F_{q}^{\sigma}$.

\State \label{recover} Recover $\restrict F{\vP_j,\vU_i}$ from
$(y_{\vR_{i,1}},\dots,y_{\vR_{i,q-1}})$ using a univariate decoding
algorithm on the values $(h_{i}(\alpha_b))(e)$ defined
in~\refeq{eq:h_i:alpha_n:e}.
\EndFor

\State\label{solve}Solve for the indeterminates $\Hasse F {\vv}(\vP_j)$,
$\size\vv<s$, the linear system~\ref{eq:erroneous:lin}.

\State\Return 
$
\left\{\Hasse F {\vv}(\vP_j), \;\size{\vv}<s\right\}=\ev^s_{\vP_j}(F)
$.
\end{algorithmic}
\caption{Local decoding algorithm for Multiplicity Codes}
\label{algo:localMRM}
\end{algorithm}

\begin{figure}[h]
\begin{center}
\includegraphics[height=.22\textwidth]{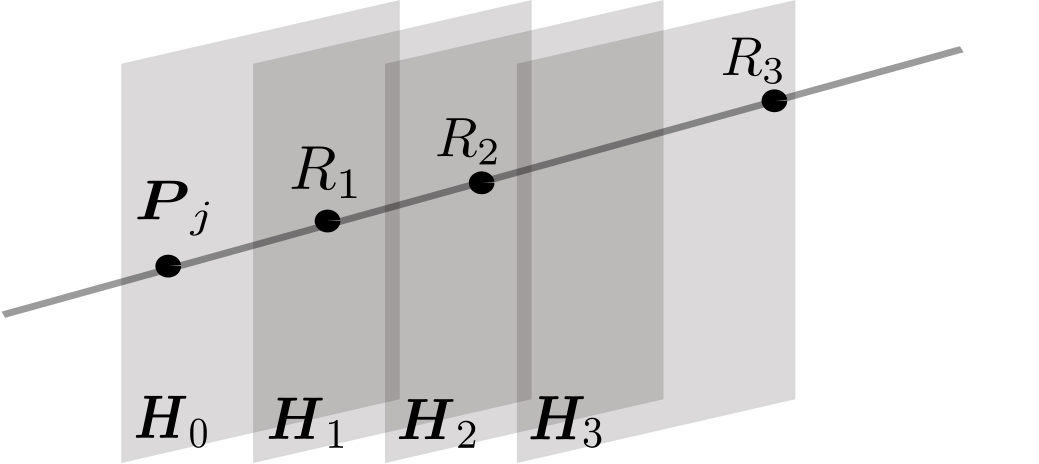}
\includegraphics[height=.22\textwidth]{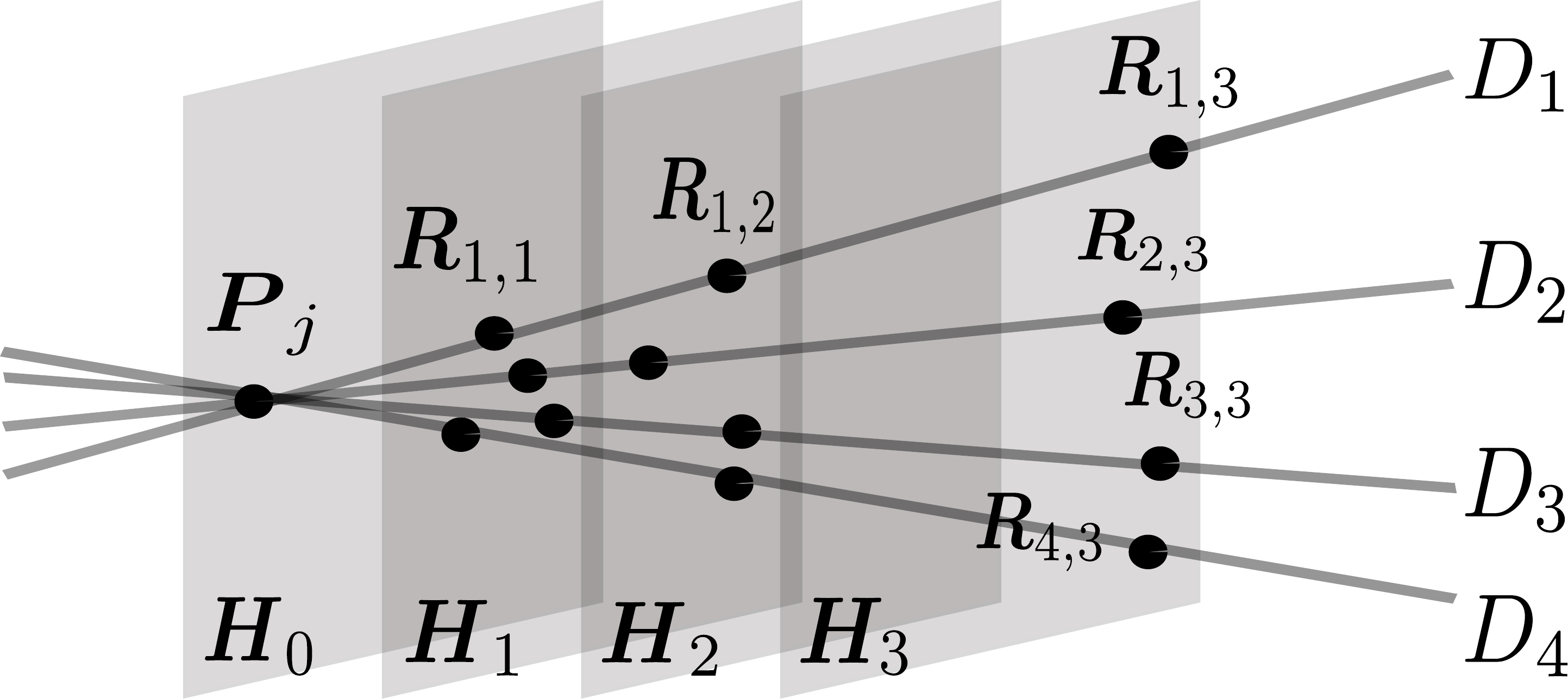}
\end{center}
\caption{Transversal lines for simple Reed-Muller codes (a), for
  Multiplicity Codes (b), assuming that the point $\vP_j$
  corresponding to query $j$ lies on the $H_0$ hyperplane. Parameters
  are $q=4$, $m=3$, $s=2$, $\sigma=4$. Not all point names are displayed
  for readability.}
\label{fig:transversal}
\end{figure}

\section{Hyperplane partitions and their use in PIRs} \label{hyp}

\subsection{Affine hyperplanes and servers}
Considering $\MRM^s_d$, we show how to equally share a codeword
\[
c=\ev^s(f)=\vect{\ev^s_{\vP_1}(f),\dots,\ev^s_{\vP_n}(f)}
\]
on $\ell=q$ servers, using the geometry of $\F_q^m$. This is done as
follows: consider $H$ a $\F_q$-linear subspace of $\F_q^m$ of
dimension $m-1$. It can be seen as the kernel of a linear map
\[
f_H:\begin{array}[t]{rcl}
\F_q^m&\rightarrow& \F_q\\
(x_1,\dots,x_m)&\mapsto&h_1x_1+\cdots h_mx_m
\end{array}
\]
for some $(h_1,\dots,h_m)\in\F_q^m\setminus\set{0}$.  Now $\F_q^m$ can
be split as the disjoint union of affine hyperplanes $ \F_q^m=H_0\cup
H_1\cup\cdots \cup H_{q-1} $, where
\[
H_i=\set{\vP\in\F_q^m\;\mid \; f_H(\vP)=\alpha_i},\quad i=0,\dots,q-1.
\]
As a simple example, consider the $\F_q$-linear hyperplane $H$ of
$\F_q^m$: 
\[H=\set{P=\vect{x_1,\dots,x_m}\;\mid x_m=0}.
\]
Then we have
$\F_q^m=H_0\cup H_1\cup\dots H_{q-1}$ where
\[
H_i=\set{\vP=\vect{x_1,\dots,x_m}\in\F_q^m\;\mid x_m=\alpha_i}, \quad
i=0,\dots,q-1.
\]
Up to a permutation of the indices, we can write any codeword 
\linebreak$
c=\vect{c_{H_0}|\cdots|c_{H_{q-1}}}
$,
where
\[
c_{H_i}=\vect{\ev^s_{\vP}(f)}_{\vP \in H_i}, \; i=0\dots,q-1.
\]
Now consider an affine line, which is \emph{transversal} to all the
hyperplanes. It is a line which can be given by any direction
$\vU\in\F_q^m\setminus\set0$ such that $f_H(\vU)\neq0$, and which
contains a point $\vP$:
\[
D=\set{\vP+t\cdot \vU\; \mid t\in \F_q}.
\]
In other words, it is  a line not contained in any
of the hyperplane $H_0,\dots,H_{q-1}$. Then, 
\[
D\cap H_j=\set{\vQ_j}, \quad j=0,\dots,q-1,
\]
for some points $\vQ_0,\dots,\vQ_{q-1}$. Now, as long as $\vU_i$,
$i=1,\dots,\sigma$, does not belong to $H$,
Algorithm~\ref{algo:localMRM} works, using the points
$\{\vQ_{i,j}\}_{0 \leq j \leq q-1}$, where $D_i\cap
H_j=\{\vQ_{i,j}\}$, $D_i$ being the line with direction $\vU_i$
passing through $\vP_j$, one query being a fake one (see
section~\ref{LDCinPIR} below).
\begin{figure}[h]
\includegraphics[height=.22\textwidth]{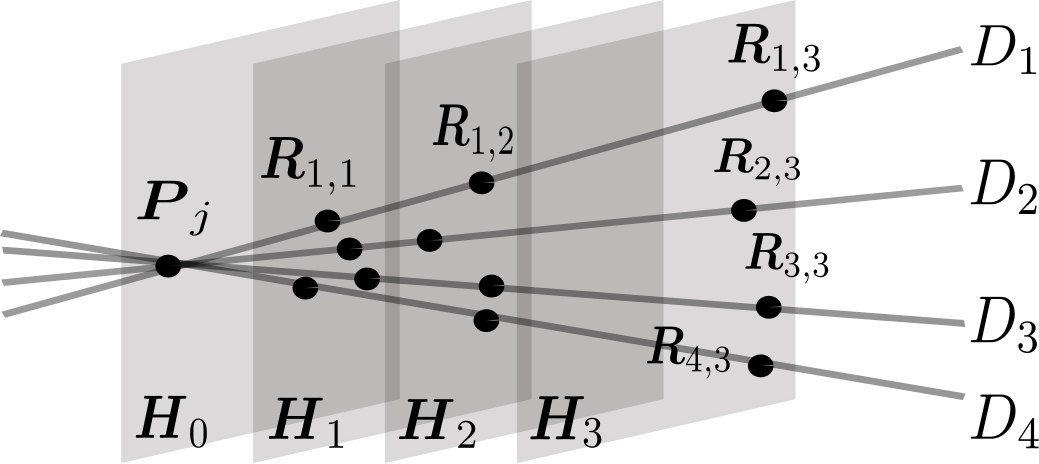}
\includegraphics[height=.22\textwidth]{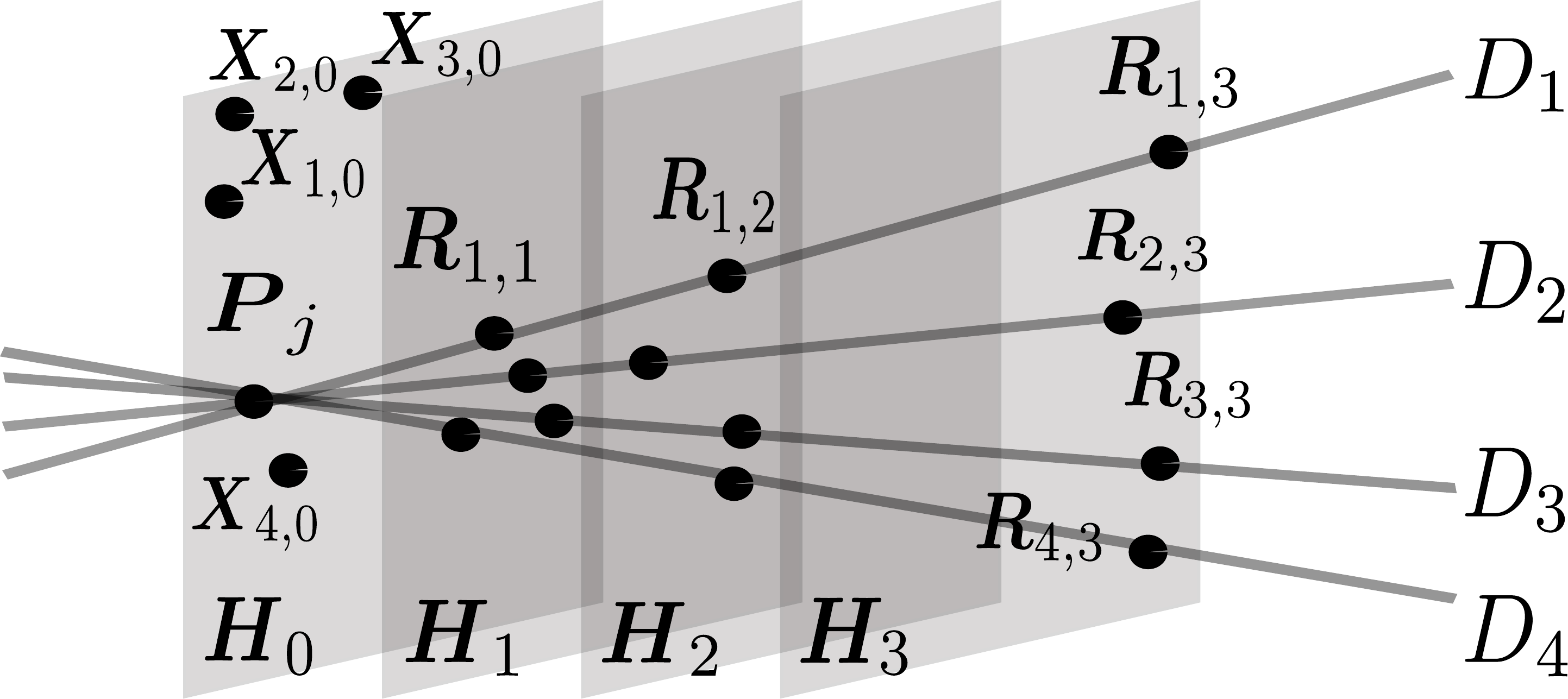}
\caption{Parameters are $q=4$, $m=3$, $s=2$, $\sigma=4$. Queries for a
  Multiplicity code used as an LDC codes (a), used in PIR
  scheme (b), assuming that the point $\vP_j$ corresponding to query
  $j$ lies on the $H_0$ hyperplane.  In the PIR scheme, random points
  $\vX_{1,0},\dots,\vX_{4,0}$ are sent to the server $S_0$ to hide 
  him the fact that he hosts the index of the request. Not all
  point names are displayed for readability.}
\label{fig:ldc-in-pir}
\end{figure}

\subsection{Use in PIR schemes} \label{LDCinPIR} Given $\F_q^m=H_0\cup
H_1\cup\dots H_{q-1}$, the PIR scheme can be built by requiring that,
for $i=1,\dots,q$, Server $S_i$ is given $c_{H_i}$ to store. Local decoding must be done using transversal lines. The user will
first select $\sigma$ transversal lines $D_i$, $i=1,\dots,\sigma$,
which passes through the point $\vP_j$ which corresponds to the
requested symbol, and query each server $S_i$ at the point $D\cap
H_i$.  In algorithms~\ref{algo:localMRM}, \ref{pick_lines}, the main
and only change is to make sure that all lines under consideration are
indeed transversal to the chosen hyperplanes.   We here explain how
this works: the code requires $(q-1)$ queries along each line. In our
context, when $\vP_j$ is requested, all $\sigma$ lines have to pass
through $\vP_j$. For a direction $\vU_i$, the queries sent to the
servers correspond to $q-1$ points on the line $D_i$ defined by
$\vU_i$, those points being all different from $\vP_j$. Assume for
instance that $\vP_j=(x_1,\dots,x_m)$ with $x_m=\alpha_u$, for some
$u$. Query $\vP_j$ must not be sent to server $S_u$ who stores the
$c_{H_u}$ part of the encoded word: $S_u$ would then know that it has
the index of the requested coordinate among its possibly queried
indices.  A solution to this problem is to send $\sigma$ fake
(i.e. random) queries $X_{i,u}$, $i=1,\dots,\sigma$, to server $S_u$,
see Fig~\ref{fig:ldc-in-pir}. This is enough to obfuscate server
$S_u$.  See Algorithm~\ref{algo:PIRtransverse}.

\begin{algorithm}[h]
\begin{algorithmic}[1]
  \renewcommand{\algorithmicrequire}{\textbf{Preprocessing Phase:}} \Require
  The user: \State chooses $q,m,d,s$ so that the original data $x$ of
  bit-size $k$ can be encoded using $\MRM^s_d(q)$, i.e.\ parameters
  such that $\binom{m+d}{d}\log_2q \geq k$; 

  \State encodes the data $x$ into the codeword $c=\ev^s(F)$, where
  the coefficients of $F$ represent the original data $x$; 

  \State sends each server $S_{\ell}$ the $c_{H_{\ell}}$ part of the
  codeword.  \renewcommand{\algorithmicrequire}{\textbf{Online
      Protocol:}}

  \Require To recover $c_{j}=ev^s_{\vP_j}(F)$ for an index $j \in [n]$, the
  user:

  \State User selects $\sigma$ distinct lines $D_i, 1 \leq i \leq
  \sigma$,  transversal to the hyperplanes, and passing through
  $\vP_j$;

  \State Let $\ell_j$ be such that $D_i \cap H_{\ell_j}=\vP_j,$
  $i=1,\dots,\sigma$

  \State For $1 \leq \ell \leq q$, $\ell \neq \ell_j$, user sends the
  queries $\{D_i \cap H_{\ell}=R_{i,\ell}\}_{1 \leq i \leq \sigma}$ to
  server $\ell$. 

  \State User sends $\sigma$ random queries $X_{i,\ell_j}$,
  $i=1,\dots,\sigma$ to server $S_{\ell_j}$; \State For $1 \leq \ell \leq q$,
  server sends the answers $\{{y_{R_{i,\ell}}}\}_{1 \leq i \leq
    \sigma}$. Answers $\{y_{R_{i,u}}\}_{1 \leq i \leq \sigma}$ are discarded by the user. 
\State User then proceeds as in steps \ref{identify} to
  \ref{solve} of algorithm~\ref{algo:localMRM} to retrieve
  $\ev^s_{\vP_j}(F)$.

\end{algorithmic}
\caption{PIR Protocol from transversal lines on hyperplanes}
\label{algo:PIRtransverse}
\end{algorithm}

\section{Analysis of the protocol given in algorithm \ref{algo:PIRtransverse}}

\subsection{Overall storage overhead}\label{subsec:overhead}
The natural reduction from locally decodable codes to information
theoretically secure private information retrieval schemes leads to
two overheads: the first one is $1/R$ where $R$ is the rate of the
code used for encoding the data, the second one is $\ell$, where
$\ell$ is the number of servers. The total overhead is thus
$\ell\cdot1/R$. Our scheme has an overhead of only $1/R$, which is the
natural overhead of the code.  With respect to the amount of storage
required in each server for encoding $k$ symbols, only $k/Rq$ symbols
are required per server. In particular, when $R\geq 1/q$, each server
stores less than $k$ symbols, which is the amount of information
without redundancy. 

\subsection{Communication complexity}
We count the communication complexity in terms of the number of
exchanged bits during the online protocol, discounting the
preprocessing phase.  The user has to send $\sigma$ points to each
server $S_j$, $j=1,\dots,q$. A point consists in $m$ coordinates in
$\F_q$, but since it belongs to an hyperplane, it can be specified
with $(m-1)$ coordinates, i.e.\ $(m-1)\log_2q$ bits. Thus
$\sigma(m-1)\log_2 q$ bits are sent to each server $S_j$, $1\leq j\leq
q$, for a total of $q\sigma(m-1)\log_2 q$.  For his response, each
server sends $\sigma$ field elements for each of the $\sigma$ points
it receives in the query: $\sigma^2$ field elements, i.e.\
$\sigma^2\log_2q$ bits, and thus a total of $q\sigma^2\log_2q$ bits
for all the servers. The overall communication complexity for the
queries and the answers is $q\sigma(m-1)\log_2
q+q\sigma^2\log_2q=(m-1+\sigma)q\sigma\log_2q =O(q\sigma^2\log_2q)$,
since $m\leq \sigma$ as soon as $s> 1$.  





\begin{figure}
\begin{center}
\begin{tabular}{|c|c|c|c|c|c|c|c|c|c|c|}
\hline
\multicolumn{5}{|c|}{Parameters}&\multicolumn{2}{c|}{Locality}&\multicolumn{2}{c|}{Storage
overhead}&\multicolumn{2}{c|}{Comm.\ complexity}\\
 \hline 
$q$ & $m$ & $s$ & $d$ &$k$ & $\sharp$ queries & $\sharp$ servers & std & ours & std & ours\\
 \hline\hline\ 16 & 2 & 1 & 14 & 120 & 15 & 16 & 32 & 2.1 & 180 & 128 \\
 16 & 2 & 2 & 29 & 465 & 45 & 16 & 25 & 1.7 & 900 & 768 \\
 16 & 2 & 3 & 44 & 1035 & 90 & 16 & 22 & 1.5 & 2880 & 2688 \\
 16 & 2 & 4 & 59 & 1830 & 150 & 16 & 21 & 1.4 & 7200 & 7040 \\
 16 & 2 & 5 & 74 & 2850 & 225 & 16 & 20 & 1.3 & 15300 & 15360 \\
 16 & 2 & 6 & 89 & 4095 & 315 & 16 & 20 & 1.3 & 28980 & 29568 \\
 \hline 
 16 & 3 & 1 & 14 & 680 & 15 &16 & 90 & 6.0 & 240 & 192 \\
 16 & 3 & 2 & 29 & 4960 & 60 & 16 & 50 & 3.3 & 1680 & 1536 \\
 16 & 3 & 3 & 44 & 16215 & 150 & 16 & 38 & 2.5 & 7800 & 7680 \\
 16 & 3 & 4 & 59 & 37820 & 300 & 16 & 32 & 2.2 & 27600& 28160 \\
 16 & 3 & 5 & 74 & 73150 & 525 & 16 & 29 & 2.0 & 79800 &82880 \\
 16 & 3 & 6 & 89 & 125580 & 840 & 16 & 27 & 1.8 & 198240 &207872 \\
 \hline 
 16 & 4 & 1 & 14 & 3060 & 15 & 16 & 320 & 21 & 300 &256 \\
 16 & 4 & 2 & 29 & 40920 & 75 & 16 & 120 & 8.0 & 2700 & 2560 \\
 16 & 4 & 3 & 44 & 194580 & 225 & 16 & 76 & 5.1 & 17100 & 17280 \\
 16 & 4 & 4 & 59 & 595665 & 525 & 16 & 58 & 3.9 & 81900 & 85120 \\
 16 & 4 & 5 & 74 & 1426425 & 1050 & 16 & 48 & 3.2 & 310800 & 327040 \\
 \rowcolor{grey}16 & 4 & 6 & 89 & 2919735 & 1890 & 16 & 42 & 2.8 & 982800 & 1040256 \\
 \hline \hline 
 256 & 2 & 1 & 254 & 32640 & 255 & 256 & 510 & 2.0 & 6120 & 4096 \\
 256 & 2 & 2 & 509 & 130305 & 765 & 256 & 380 & 1.5 & 30600 & 24576 \\
 256 & 2 & 3 & 764 & 292995 & 1530 & 256 & 340 & 1.3 & 97920 & 86016 \\
 256 & 2 & 4 & 1019 & 520710 & 2550 & 256 & 320 & 1.3 & 244800 & 225280 \\
 256 & 2 & 5 & 1274 & 813450 & 3825 & 256 & 310 & 1.2 & 520200 & 491520 \\
 256 & 2 & 6 & 1529 & 1171215 & 5355 &256 & 300 & 1.2 & 985320 & 946176 \\
 \hline 
 \rowcolor{grey}256 & 3 & 1 & 254 & 2796160 & 255 & 256 & 1500 & 6.0 & 8160 & 6144 \\
 256 & 3 & 2 & 509 & 22238720 & 1020 & 256 & 770 & 3.0 & 57120 & 49152 \\
 256 & 3 & 3 & 764 & 74909055 & 2550 & 256 & 570 & 2.2 & 265200 & 245760 \\
 256 & 3 & 4 & 1019 & 177388540 & 5100 & 256 & 480 & 1.9 & 938400 & 901120 \\
 256 & 3 & 5 & 1274 & 346258550 & 8925 & 256 & 430 & 1.7 & 2713200 & 2652160 \\
 256 & 3 & 6 & 1529 & 598100460 & 14280 & 256 & 400 & 1.6 & 6740160 & 6651904 \\
 \hline 
 256 & 4 & 1 & 254 & 180352320 & 255 & 256 & 6100 & 24 & 10200 & 8192 \\
 256 & 4 & 2 & 509 & 2852115840 & 1275 & 256 & 1900 & 7.5 & 91800 & 81920 \\
 256 & 4 & 3 & 764 & 14382538560 & 3825 & 256 & 1100 & 4.5 & 581400 & 552960 \\
 256 & 4 & 4 & 1019 & 45367119105 & 8925 & 256 & 840 & 3.3 & 2784600 & 2723840 \\
 256 & 4 & 5 & 1274 & 110629606725 & 17850 & 256 & 690 & 2.7 & 10567200 & 10465280 \\
 256 & 4 & 6 & 1529 & 229222001295 & 32130 & 256 & 600 & 2.4 & 33415200 & 33288192 \\
 \hline
\end{tabular}
\end{center}
\caption{Properties of our scheme for $q=16$ and $q=256$. We have to
distinguish LDC-locality (i.e. $\sharp$ queries) and PIR-locality (i.e. $\sharp$ servers), since they are not the same
using our construction. The storage overhead is the global overhead
among all the servers: in the standard case, using the standard LDC to
PIR reduction as in Lemma~\ref{lemma:LDC2PIR}, it is $(q-1)/R$; in our
case, using partitioning on the servers, it is $1/R$, $R$ being the
rate of the code. Similarly the communication complexities (in bits)
are shown. The degree $d$ has been chosen to be $d=s(q-1)-1$, the
maximum possible value, with no correction capability.}
\label{fig:params}
\end{figure}

\subsection{PIR-locality} Our construction leads to introduce the
notion of ``PIR-locality'': when an LDC code admits a nice layout as multiplicity codes do, the
number of servers can be smaller than the locality of the code.  We
call this the \emph{PIR-locality}. Here the (LDC-)locality,
i.e. the number of queries, is $(q-1)\sigma$, while the PIR-locality, i.e. the number of servers, is $q$. 
The tables show the obtained parameters for $q=256$ and $q=16$
in~\reffig{fig:params}.  We can see that the rate and LDC-locality of the
code grow with $s$, while the PIR-locality is constant for a fixed
$q$. The global storage overhead is much smaller, and the
communication complexities are very similar.

\subsection{Robustness of the protocol}
Algorithm~\ref{algo:localMRM} involves $\sigma$ applications of
decoding of univariate multiplicity codes of length
$q-1$. From~\cite{KopSarYek2011}, we can decode if the  word
$y_i=(y_{R_{i,1}},\dots,y_{R_{i,q-1}})$ is $t$-far from a codeword
${\ev^s(F)}$, for a polynomial $F \in \F_q[X_1]_d$, where
$t=1/2(q-1-d/s)$.  The received word $y_i$ corresponds to the answers
of the $q-1$ servers (all $q$ servers except server $u$) for direction
$\vU_i$. Tolerating $\nu=\lfloor t\rfloor$ errors here means that
$\nu$ servers can answer wrongly. Thus, following the terminology of
Beimel and Stahl~\cite{JOC07}, our protocol is a $\nu$-Byzantine
robust protocol.

We sum up features of the protocol presented in
Algorithm~\ref{algo:PIRtransverse} in the following

\begin{theorem}
  Let $q$ be a power of a prime, $m, s \in \N^*$, and $d$ be an
  integer with $d<s(q-1)$. Set $\sigma=\binom{m+s-1}{m}$, with the
  constraint $\sigma \leq (q^m-1)/(q-1)$.  Protocol from
  Algorithm~\ref{algo:PIRtransverse} has:
\begin{itemize}
\item LDC-locality (i.e. number of queries) $\sigma(q-1)$;
\item PIR-locality (i.e. number of servers) $\ell=q$;
\item Communication complexity $(m-1+\sigma)q\sigma\log_2q$ bits;
\item Storage overhead $1/R$, where $R=\binom{m+d}{m}/(q^m\sigma)$ is
  the rate of the underlying multiplicity code;
\item $\nu$-Byzantine robustness, where $\nu=\lfloor 1/2(q-1-d/s)\rfloor$, in the sense that it can tolerate up to $\nu$ servers answering wrongly.
\end{itemize}

\end{theorem}

\section{Discussing parameters}
\subsection{Impact of the Byzantine robustness on the storage
  overhead}
Expressing $d$ in terms of $t$ for a given $s$ gives $d=(q-1)s-2st$,
which then gives a rate, say $R_t$, to be compared with the rate $R$
found for $d=s(q-1)-1$, when no error can be tolerated. For small $m$ and
$s(q-1)$ large enough, we have a relative loss:
\begin{align*}
  R_t/R = \frac{\binom{s(q-1)-2st+m}{m}/\sigma q^m }{ \binom{s(q-1)+m-1}{m}/\sigma q^m}
 & \approx \left( \frac{(q-1-2t +m/s) }{ (q-1+(m-1)/s) }\right) ^m
\end{align*}
For $m=s=1$, we find $(q-2t)/(q-1)$, which is almost the rate of the $t$-error
correcting classical Reed-Solomon code. Otherwise, we get, for small $t$
\[
R_t/R \approx \left(1-\frac{2t-1/s}{q-1+(m-1)/s}\right)^m 
\]
For $t=1$ or 2 and $m$ small, the relative loss is not drastic.  But, if $t$ is
large, say $(q-1)/2$
\[
R_t/R  \approx (1/(sq))^m
\]
and the loss is bigger. 

\subsection{Choice of $q$} We discuss how the size of $q$ may be
chosen independently of the size of entries on the database.  Consider
a simple database, which is a table, with $E$ entries, each entry
having $S$ records, all of the same bit-size $b$.  I.e.\ the
total bit-size of the database is thus $N=E\cdot S\cdot b$. A
multiplicity code of $\F_q$-dimension $k$ enables to encode $k\log_2
q$ bits. Thus, to encode the whole database, we need $k\log_2 q\geq
N=E\cdot S\cdot b$.  If furthermore $a=b/\log_2 q$ is an integer,
then, to recover a record of size $b$, the user needs to apply the PIR
protocol $a$ times.  By definition of information theoretic PIR
schemes, Protocol.~\ref{algo:PIRtransverse} can be run any number of
times, with no information leakage. This implies that $q$ does not
need to have a special relationship with the original data.

For instance, imagine a database of 90\,000 IPV6 adresses. An IPV6
address consists in 128 bits addresses, i.e.\ 16 bytes. The database
has $E=90\,000$, $S=1$, $b=128$, and requires $90\,000\cdot 16= 144\,
0000$ bytes of storage.  We first design a PIR scheme using
$q=256=2^8$. Mapping a byte to an $\F_q$-symbol, we need a code of
$\F_q$-dimension at least $144\, 000$. From Table~\ref{fig:params},
using $m=3$, $s=1$, we find a code of $\F_q$-dimension
$2796160\sim2,7\cdot 10^6$, and expansion $6$. The LDC-locality is
255, and its PIR-locality is 256. The communication cost is 6144 bits.

But we could also use $q_0=2^4=16$. Then $144\,0000$
bytes require  $2\cdot144\, 0000=2,88\cdot10^6$
$\F_{q_0}$-symbols. From Table~\ref{fig:params}, with $m=4$ and $s=6$, we
find a code of $\F_{q_0}$-dimension $2919735\sim2.9\cdot 10^6$, and
expansion 2.8. Its LDC-locality is
$(q_0-1)\binom{4+6-1}{4}=15\cdot126=1890$ while its PIR-locality is
 $16$. This is better in many aspects
since less servers are needed, and a better rate is achieved. But the
communication cost is now 1040256 bits.

\section{Conclusion}
Starting from multiplicity codes, we have designed a layout of the
encoded data which leads to a new PIR scheme. It features a very small
PIR-locality and much smaller global
redundancy compared to PIR schemes naturally arising from LDCs, as
well as Byzantine robustness.  This layout is quite natural in the
context of multiplicity codes. A straightforward question, to be
investigated in a future work, is to construct layouts for other
locally decodable codes, like affine-invariant codes~\cite{GKS13} and
matching vector codes~\cite{Yekhanin08,Efremenko09}. This seems
feasible due to the very multidimensional and geometric nature of
these constructions.

\bibliographystyle{plain}
\bibliography{ldc}





\appendix

\section{Possible ranges for $d$} \label{dsq-1}
In order the encoding function $\ev^s$ to be injective, it is sufficient to choose $d<sq$. Indeed:
\[
\ev^s(f)=\ev^s(g) \Leftrightarrow \ev^s(f-g)=(0,\ldots,0), 
\]
which means that $f-g$ admits $sq^m$ zeroes, counting multiplicities. By Schwartz-Zippel lemma, we have:
\[
\sum_{P \in \F_q^m} {\rm mult}(f-g,P) \leq d q^{m-1}\]
that is here
\[
sq^m \leq dq^{m-1}
\]
Thus, if we want $f-g$ to be identically zero, it suffices that $d <sq$.

Now during the decoding phase, in the case of errors, one has to
perform Reed-Solomon with multiplicities decoding (indeed, $\sigma$
Reed-Solomon applications of decoding). In this case, the length of the Reed-Solomon
code is always $q-1$ as we have $q-1$ noisy evaluations of the
original polynomial $F$ on each line. In order such a Reed-Solomon
code to realize proper (i.e. injective) encoding, we need $d <
s(q-1)$, as shown below.
\[
\ev^s(f)=\ev^s(g) \Leftrightarrow \ev^s(f-g)=(0,\ldots,0), 
\]
i.e. $f-g$ admits $s(q-1)$ zeroes, where here $\ev^s(f)$ is the
encoding of a univariate degree $\leq d$ polynomial $f \in \F_q[X]$
with a Reed-Solomon code of length $q-1$ and multiplicity $s$. But a
univariate polynomial cannot have more zeroes, counted with
multiplicities, than its degree:
\[
\sum_{P \in \F_q^*} {\rm mult}(f-g,P) \leq d,\]
thus if we want $f-g$ to be identically zero, it suffices that $d < s(q-1)$.

\section{Decoding univariate multiplicity codes} \label{univ-mult}
When the number $m$ of variables is $1$, then the codes lead to
Reed-Solomon codes, also called derivative codes
in~\cite{GuruWang2013}. We briefly recall how to decode these codes,
using the so-called Berlekamp-Welch
framework~\cite{Gemmell-Sudan:IPL1992}. We consider univariate
polynomials in $\F_q[X]$. For $s>0$, we have $\Sigma=\F_q^s$, and
the code is the set of codewords of length $n=q-1$:
\[
\set{c=\ev^s(F)\mid\; F\in\F_q[X]_d}.
\]
Decoding up to distance $t$ is, for a given vector $y\in\Sigma^n$,
find all polynomials $F\in\F_q[X]_d$ such that
\[
d_\Sigma(\ev^s(F),y)\leq t
\]
where $d_\Sigma$ is the Hamming distance in $\Sigma^n$. We first look
for two polynomials $N,E\in\F_q[X]$ of degree $(sn+d)/2$ and
$(sn-d)/2$ respectively, as follows. Write the linear system of equations:
\[
\left\{
\begin{array}{rcl}
N(\alpha_i)&=&E(\alpha_i)\cdot y_{i,0}\\
\Hasse N 1(\alpha_i)&=&E(\alpha_i)y_{i,1}+\Hasse E 1 (\alpha_i) \cdot y_{i,0}\\
&\vdots\\
\Hasse N {s-1}(\alpha_i)&=&\sum_{j=0}^{s-1}  \Hasse E j(\alpha_i)\cdot y_{i,s-1-j}
\end{array}
\right.
\]
for $ i=1,\dots,n$, where the indeterminates are the coefficients of
$N$ and $E$.  This is a system of $sn$ homogeneous linear equations in
$(sn-d)/2+1+(sn+d)/2+1=sn+2$ unknowns. Thus a non-zero solution
$(N,E)$ always exists. Given any solution, $F$ can then be recovered
as $N/E$.
 
Assuming that $t=(n-d/s)/2$, we can show the correctness of this
algorithm: any univariate polynomial $F$ of degree $\leq d$, such
that $d_\Sigma(\ev^s(F),y)\leq t$ will satisfy $N-EF=0$ where $(N,E)$
is a solution of the above system. Indeed, for any $\alpha_u$ such
that $\ev^s_{\alpha_u}(F)=y_u$, the system is satisfied at $\alpha_u$,
and hence the polynomial $N-EF$ has a zero of multiplicity $s$ at
$\alpha_u$. Thus
\[
\sum_{i=1,\ldots,n} {\rm mult}(N-EF,\alpha_i) > (n-t)s=(sn+d)/2.
\]
But ${\rm deg}(N-EF) \leq {\rm max}\{(sn+d)/2,d+(sn-d)/2\}=(sn+d)/2$. Thus $N-EF$, having more zeroes than its degree, is identically zero.

\end{document}